\documentclass[submission,copyright,creativecommons]{eptcs}
\providecommand{\event}{Non-Classical Logics, Theory, and Applications 2022 ($NCL'22$)} 
\usepackage{breakurl}    
\usepackage{underscore}    
\usepackage{proof}
\usepackage{toyooka_logic}
\newcommand{\imp}{\mathbin{\to}}
\theoremstyle{definition}
\newtheorem{definition}{Definition}[]
\newtheorem{corollary}{Corollary}[]

\title{Combining First-Order Classical and Intuitionistic Logic \\ }
\author{Masanobu Toyooka
\institute{Hokkaido University\\ Sapporo, Japan}
\institute{Graduate School of Letters}
\email{toyooka.masanobu.t1@elms.hokudai.ac.jp}
\and
Katsuhiko Sano
\institute{Hokkaido University\\
Sapporo, Japan}
\institute{Faculty of Humanities and Human Science}
\email{v-sano@let.hokudai.ac.jp}
}
\def\titlerunning{Combining First-Order Classical and Intuitionistic Logic}
\def\authorrunning{M.Toyooka 
\& K.Sano}
\newcommand{\ljc}{\mathsf{G}(\mathbf{C}+\mathbf{J})}

\newcommand{\fljc}{\mathsf{G}(\mathbf{FOC}+\mathbf{J})}
\newcommand{\fljcc}{\mathsf{G^{-}}(\mathbf{FOC}+\mathbf{J})}

\begin{document}

\maketitle

\begin{abstract}
This paper studies a first-order expansion of a combination $\mathbf{C+J}$ of intuitionistic and classical propositional logic, which was studied by Humberstone (1979) and del Cerro and Herzig (1996), from a proof-theoretic viewpoint. 
While $\mathbf{C+J}$ has both classical and intuitionistic implications, our first-order expansion adds
classical and intuitionistic universal quantifiers and one existential quantifier to $\mathbf{C+J}$. 
This paper provides a multi-succedent sequent calculus $\fljc$ for our combination of the first-order intuitionistic and classical logic. 
Our sequent calculus $\fljc$ restricts contexts of the right rules for intuitionistic implication and intuitionistic universal quantifier to particular forms of formulas. 
The cut-elimination theorem is established to ensure the subformula property. As a corollary, $\fljc$ is conservative over both first-order intuitionistic and classical logic. 
Strong completeness of 
$\fljc$ is proved via a canonical model argument. 
\end{abstract}

\section{Introduction}
\subsection{Introduction and Motivation}
\label{sec:int}
This paper studies a proof-theoretic aspect of a first-order expansion of a combined logic $\mathbf{C+J}$ of intuitionistic and classical propositional logic, which was studied by Humberstone~\cite{Humberstone1979} and del Cerro and Herzig~\cite{Cerro1996}. While $\mathbf{C+J}$ has both classical and intuitionistic implications, our first-order expansion adds
classical and intuitionistic universal quantifiers and one existential quantifier to $\mathbf{C+J}$. In particular, this paper proposes a cut-free sequent calculus called $\fljc$, which is an expansion of the calculus $\ljc$ in~\cite{Toyooka2021a} for the combination of intuitionistic and classical propositional logic. 

There are various semantic methods to combine intuitionistic and classical logic, such as ones in~\cite{Caleiro2007,De2013,De2014}, but this paper follows the semantic treatment in~\cite{Cerro1996,Humberstone1979}, whose main idea is adding intuitionistic and classical implications into one logic called $\mathbf{C+J}$. 
Each implication (expressed as  $\imp_{\mathtt{i}}, \imp_{\mathtt{c}}$, respectively) is interpreted in a Kripke model as follows: 
\[
\begin{array}{lll}
w \models_{M} A \imp_{\mathtt{i}} B & \iff & \text{for all } v \in W, (wRv \text{ and } v \models_{M} A \text{ jointly imply } v \models_{M} B),\\
w \models_{M} A \imp_{\mathtt{c}} B & \iff & w \models_{M} A \text{ implies } w \models_{M} B,\\ 
\end{array}
\]
where $M$ is an intuitionistic Kripke model, $w$ is a possible world in $M$, and $R$ is a preorder equipped in $M$.

It is well-known that an intuitionistically valid formula $A \imp_{\mathtt{i}} (B \imp_{\mathtt{i}} A)$ corresponds to the property called {\em heredity} with respect to $A$ in intuitionistic Kripke semantics, which is defined as: $w\models A$ and $wRv$ jointly imply $v \models A$ for all Kripke models $M$ and all states $w$ and $v$ in $M$\footnote{The correspondence between heredity and the formula $A \imp_{\mathtt{i}} (B \imp_{\mathtt{i}} A)$ is mentioned in~\cite{Jongh2018,Jongh2015,Restall1994}}. However, the existence of classical implication breaks this heredity in the Kripke semantics of $\mathbf{C + J}$, i.e., there is a Kripke model where $\neg_{\mathtt{c}} p \imp_{\mathtt{i}} (\top \imp_{\mathtt{i}} \neg_{\mathtt{c}}p)$ is not valid. 
Because of this semantic phenomenon, del Cerro and Herzig~\cite{Cerro1996} added to their Hilbert system for $\mathbf{C + J}$ the following syntactically restricted axiom:
\[
\texttt{(PER)} A \imp_{\mathtt{c}} (B \imp_{\mathtt{i}} A)^{\dag} \quad \dag \text { $A$ is a persistent formula}, 
\]
where a persistent formula can be understood  essentially as a formula of the form atomic $p$ or $A \imp_{\mathtt{i}} B$ ($A$ and $B$ possibly contain the classical implication)\footnote{The reason why we use the expression ``essentially'' is the definition provided in~\cite{Cerro1996} is slightly different.}.

By employing the idea of the axiom $\texttt{(PER)}$, a cut-free sequent calculus $\ljc$ of $\mathbf{C+J}$ was proposed in~\cite{Toyooka2021a} where the right rule for the intuitionistic implication is restricted as follows:
\[
\infer[(\Rightarrow \imp_{\mathtt{i}})]{C_{1} \imp_{\mathtt{i}} D_{1}, \ldots, C_{m} \imp_{\mathtt{i}} D_{m}, p_{1}, \ldots p_{n} \Rightarrow A \imp_{\mathtt{i}} B}{C_{1} \imp_{\mathtt{i}} D_{1}, \ldots, C_{m} \imp_{\mathtt{i}} D_{m}, p_{1}, \ldots p_{n}, A \Rightarrow B},
\]
while the ordinary right rule for the intuitionistic implication is of the following form:
\[
\infer{\Gamma \Rightarrow A \to_{\mathtt{i}} B}{\Gamma, A \Rightarrow B}.
\]

\noindent  Even though this restriction is imposed, a sequent derivable in intuitionistic logic is also derivable in $\ljc$. This is because the ordinary rule is derivable in $\ljc$ if the context of the rule contains only intuitionistic formulas. The ordinary rule is also derivable in the calculus obtained from intuitionistic propositional sequent calculus by replacing it with the restricted one. Therefore, the restricted version of the rule captures the core of the original one. 

When expanding $\mathbf{C+J}$ to a first-order syntax, we add intuitionistic and classical universal quantifiers (expressed as $\forall_{\mathtt{i}}, \forall_{\mathtt{c}}$, respectively) and one existential quantifier (expressed as $\exists$). Each universal quantifier is interpreted in a Kripke model as follows:
\[
\begin{array}{lll}
w \models_{M} \forall_{\mathtt{i}}x A & \iff & \text{ for all } v \in W, (wRv \text{ implies } \text{for all } d \in D(v), v \models_{M} A[\underline{d} / x]),
\\
w \models_{M} \forall_{\mathtt{c}}x A & \iff & \text{ for all } d \in D(w), w \models_{M} A[\underline{d} / x],\\
\end{array}
\]
where $M$ is a Kripke model for first-order intuitionistic logic and $\underline{d}$ is a syntactic name of an element $d$ in a domain. As a Kripke model for propositional intuitionistic logic, $w$ and $v$ are possible worlds in $M$ and $R$ is a preorder equipped in $M$.

Similar to the classical implication, the classical universal quantifier also breaks the heredity in a Kripke semantics, while the intuitionistic one does not. Therefore, we have to regard a formula of the form $\forall_{\mathtt{i}}x A$ as a persistent formula. This consideration enables us to expand $\ljc$ naturally to $\fljc$. 
In $\fljc$, the right rule for the intuitionistic implication can be modified to cover the new notion of persistent formula in the first-order syntax. 
Moreover, the right rule for the intuitionistic universal quantifier is defined as follows:
\[
\infer[(\Rightarrow \forall_{\mathtt{i}})]{\forall_{\mathtt{i}}x_{1} B_{1}, \ldots, \forall_{\mathtt{i}}x_{l} B_{l}, C_{1} \imp_{\mathtt{i}} D_{1}, \ldots, C_{m} \imp_{\mathtt{i}} D_{m}, p_{1}, \ldots p_{n} \Rightarrow \forall_{\mathtt{i}}x A}{\forall_{\mathtt{i}}x_{1} B_{1}, \ldots, \forall_{\mathtt{i}}x_{l} B_{l}, C_{1} \imp_{\mathtt{i}} D_{1}, \ldots, C_{m} \imp_{\mathtt{i}} D_{m}, p_{1}, \ldots p_{n} \Rightarrow A[z/x]},
\]
\noindent where $z$ must not occur free in the lower sequent. 
This paper establishes that this calculus enjoys the cut-elimination, which guarantees the subformula property, and shows that the calculus is sound and semantically complete with respect to the class of all intuitionistic Kripke models. 

This paper is structured as follows. Section \ref{sec:sks} provides our syntax, Kripke semantics for it, and the sequent calculus $\fljc$. The soundness is shown in the section. Section \ref{sec:cut} demonstrates the cut elimination theorem. Section \ref{sec:comp} establishes the strong completeness for Kripke semantics via a canonical model argument. Section \ref{sec:caf} concludes the paper and gives further direction of research. 

\subsection{Related Work}
\label{sub:rw}
An attempt related to our approach the most was the one in~\cite{Lucio2000}. In~\cite{Lucio2000}, a sequent calculus for a first-order combination of intuitionistic and classical logic, denoted by $\mathbf{FO^{\supset}}$, was given. This calculus is provided by adding intuitionistic implication to first-order classical logic, so it has only one universal quantifier (a classical one). Therefore, this study is also regarded as a different first-order expansion of $\mathbf{C+J}$ from ours. From our perspective, however, the intuitionistic universal quantifier is missing in~\cite{Lucio2000}.
In order to deal with the failure of the heredity mentioned in Section \ref{sec:int}, the calculus employs the notion of {\em structured} single succedent sequent of the form: $\Gamma_{1}; \cdots ;\Gamma_{n} \Rightarrow A$, where $\Gamma_{i}$ is a finite set of formulas.  For example, $B_{1}, B_{2}; C_{1}, C_{2}, C_{3}; D \Rightarrow A$ is a structured sequent and its semantic interpretation (see \cite[Definition 16]{Lucio2000}) is: for any $w_{1}, w_{2}, w_{3}$ such that $w_{1} \leqslant w_{2} \leqslant w_{3}$, if $w_{1} \models B_{i}$ for all $i$, $w_{2} \models C_{j}$ for all $j$ and $w_{3} \models D$, then $w_{3} \models A$. Thanks to this notion, no restriction on contexts is needed for sequents in $\mathbf{FO^{\supset}}$. It is noted that some rule for $\mathbf{FO^{\supset}}$ does not satisfy the subformula property. Lucio~\cite{Lucio2000} proved that  $\mathbf{FO^{\supset}}$ is sound and complete for the intended Kripke semantics.

Other than $\mathbf{C + J}$, there are many attempts to combine intuitionistic and classical logic. In~\cite{Caleiro2007}, a logic called $\mathbf{CIPL}$, whose syntax consists of propositional variables and intuitionistic and classical implications, was proposed, and a Hilbert system and a Kripke semantics were given. The Kripke semantics has an interpretation of classical implication different from that of $\mathbf{C+J}$. 
In~\cite{De2013,De2014}, the logic called $\mathbf{IPC}^{\sim}$ was given by adding a negation called ``empirical negation'' to intuitionistic logic. This negation may be regarded as classical negation, but the semantic interpretation of empirical negation was  different from classical one. That is, the empirical negation is evaluated at a base state of a Kripke model where a base state can see all states. 
Also,  a semantic consequence relation in~\cite{De2013,De2014} is defined in terms of base states of Kripke models.
It should be noted that all formulas satisfy heredity in $\mathbf{IPC}^{\sim}$. 
In~\cite{De2013,De2014}, a Hilbert system was given for $\mathbf{IPC}^{\sim}$. 
We, however, emphasize that all of 
\cite{Caleiro2007,De2013,De2014} stay at the propositional level. Our study of a {\em first-order} expansion of $\mathbf{C + J}$ may be useful to expand these systems to the first-order level.

Prawitz~\cite{Prawitz2015} provided a system called {\em ecumenical system} as a natural deduction system. The underlying syntax of this system is obtained by adding classical implication, disjunction, and existential quantifier to a syntax of intuitionistic logic.
An interesting syntactic feature of this system is that it has only one negation (regarded as intuitionistic negation), while two implications exist. 
A Kripke semantics and a sequent calculus for the propositional fragment were given in~\cite{Pereira2017} and~\cite{Pimentel2018} respectively, and a Kripke semantics and a sequent calculus $\mathbf{LEci}$ of the full fragment were proposed in~\cite{Elaine2019}. 
The main idea of the system is defining a classical logical connectives or quantifier in terms of intuitionistic ones. 
For example, the interpretation for the classical implication is defined as follows: $w \models A \imp_{\mathtt{c}} B$ iff $w \models \neg (A \land \neg B)$, where $\neg$ is intuitionistic negation. 
Corresponding to this interpretation, the right rule for the classical implication is defined as follows:
\[
\infer{\Gamma \Rightarrow A \imp_{\mathtt{c}} B}{\Gamma, A, \neg B \Rightarrow \bot}.
\]
\noindent It is noted that cut-elimination was shown in~\cite{Elaine2019}, but the calculus does not satisfy the subformula property as we can easily see from the above right rule of the classical implication. It is remarked that, although our syntax is different from~\cite{Pereira2017,Pimentel2018,Elaine2019,Prawitz2015}, our sequent calculus is fully analytic, i.e., all the inference rules except the rule of cut satisfy the subformula property and the calculus also enjoys the cut-elimination theorem. 

\section{Syntax, Kripke Semantics, and Sequent calculus $\fljc$}
\label{sec:sks}
\subsection{Syntax}
\label{subsec:syntax}
This section introduces the syntax for $\fljc$. Our syntax $\mathcal{L}$ consists of the following:
\begin{itemize}
    \item A countably infinite set of variables $\mathsf{Var}$ := $\{ x_{1}, x_{2}, \ldots \}$,
    \item A countably infinite set of constant symbols $ \{ c_{1}, c_{2}, \ldots \}$,
    \item A countably infinite set of predicate symbols $\{ P_{1}, P_{2}, \ldots \}$,
    \item Logical connectives: $\bot, \land, \lor, \imp_{\mathtt{i}}, \imp_{\mathtt{c}}$,
    \item Quantifiers: $\forall_{\mathtt{i}}, \forall_{\mathtt{c}}, \exists$.
\end{itemize} 

\noindent The intuitionistic and classical implications and universal quantifiers are denoted by $\imp_{\mathtt{i}}$, $\imp_{\mathtt{c}}$ and $\forall_{\mathtt{i}}$, $\forall_{\mathtt{c}}$, respectively. Only one disjunction and existential quantifier are contained, since their satisfaction relations in standard Kripke semantics are the same between classical logic and intuitionistic logic. We denote by $\mathcal{L}_{\mathbf{C}}$ (the syntax for the classical logic) and $\mathcal{L}_{\mathbf{J}}$ (the syntax for the intuitionistic logic) the resulting syntax dropping $\imp_{\mathtt{i}}$ and $\forall_{\mathtt{i}}$, $\imp_{\mathtt{c}}$ and $\forall_{\mathtt{c}}$, respectively. We define $\top$ := $\bot \imp_{\mathtt{i}} \bot$, $\neg_{\mathtt{c}} A$ := $A \imp_{\mathtt{c}} \bot$, and $\neg_{\mathtt{i}} A$ := $A \imp_{\mathtt{i}} \bot$. 

{\em Terms} consist of variables and constant symbols and are denoted by $t_{1}, t_{2}, \ldots $. A constant symbol is called a {\em closed term}, since it has no occurrence of free variables. 
The set of all formulas $\mathsf{Form}_{\mathcal{L}}$ (often written as $\mathsf{Form}$) for the syntax $\mathcal{L}$, is defined inductively as follows:
\[
A ::= P(t_{1}, \ldots, t_{n}) \,|\, \perp \,|\, A \lor A \,|\, A \land A \,|\, A \imp_{\mathtt{i}} A \,|\, A \imp_{\mathtt{c}} A \, | \, \forall_{\mathtt{i}}x A, \, | \, \forall_{\mathtt{c}}x A \, | \, \exists x A, 
\]
where $P$ denotes a predicate symbol. We denote by $\mathsf{Form}_{\mathbf{C}}$ and $\mathsf{Form}_{\mathbf{J}}$ the set of all classical formulas and the set of all intuitionistic formulas in $\mathcal{L}$, respectively.  
The set of free variables in a formula $A$ is denoted by $\mathsf{FV}(A)$. 
We define a {\em closed formula} as a formula which has no occurrence of a free variable. We employ the notion of {\em clash avoiding substitution} $[t/x]$ when we do substitution in a formula, as~\cite{Ono1994a}. By employing this notion, we can avoid the case where $t$ becomes a bound variable in the formula as an effect of substitution of $t$ for $x$. We also consider a syntax $L$ different from $\mathcal{L}$, which  contains all the logical connectives and the set of predicate and constant symbols of $\mathcal{L}$. That is, $L$ and $\mathcal{L}$ are different only in a set of variables. We denote the set of variables in a syntax $L$ by $\mathsf{Var}(L)$, and the set of all formulas in a syntax $L$ is denoted by $\mathsf{Form}_{L}$. We call a formula $A$ an {\em $L$-formula} if $A \in \mathsf{Form}_{L}$.

\subsection{Semantics}
\label{subsec:semex}
Let us move to the semantics for our syntax $\mathcal{L}$. We give a valuation and a satisfaction relation only to closed formulas, and deal with possibly non-closed formulas in the definition of semantic consequence relation (Definition \ref{def:sec}). First, we define a Kripke frame $F$ and then proceed to defining valuation $V$.

\begin{definition}
\label{def:model3}
 A {\em Kripke frame} is a tuple $F$ = $(W,R,(D(w))_{w \in W})$ where 
 \begin{itemize}
     \item $W$ is a non-empty set of possible worlds,
     \item $R$ is a preorder on $W$, i.e., $R$ satisfies reflexivity and transitivity,
     \item $D(w)$ is a non-empty set,
 \end{itemize}
 and satisfies the following:
 \begin{itemize}
     \item For all $w, v \in W, wRv$ implies $D(w) \subseteq D(v)$,
     \item ${\bigcap}_{w\in W} D(w) \neq \emptyset$.
 \end{itemize}
\end{definition} 

\begin{definition}
\label{def:model4}
A valuation $V$ on a Kripke frame $(W,R,(D(w))_{w \in W})$ is defined as the following:
\begin{description}
        \item[(rigidity for constants)] $V(c) \in {\bigcap}_{w \in W} D(w)$ ($c$ is a constant symbol),
    \item[(heredity for predicates)] $V(P,w) \subseteq D(w)^{n} $ such that for all $w, v \in W, (wRv$ implies  $V(P,w) \subseteq V(P,v))$.
\end{description}
\end{definition}

\noindent We define a quadruple $M$ = $(W,R,(D(w))_{w \in W},V)$ as a Kripke model.

Let $\mathcal{L}( A )$ be an expanded syntax of $\mathcal{L}$ by the constant symbols $\{ \underline{a} | a \in A\}$ to all the elements of $A$.

\begin{definition}
\label{def:model5}
For all closed terms $t$ of $\mathcal{L}( D(w) )$, $V(t)$ is defined as the following:
\begin{itemize}
    \item $V(t)$ = $d$ if $t \equiv \underline{d}$ and $d \in D(w)$,
    \item $V(t)$ = $V(c)$ if $t \equiv c$.
\end{itemize}
\end{definition}

\begin{definition}
\label{def:model6}
Given a model $M$ = $(W,R,(D(w))_{w \in W},V)$, a world $w \in W$ and a closed formula $A$ of $\mathcal{L}( D(w) )$, the {\em satisfaction relation} $w \models_{M} A$ is defined inductively as follows: 
\[
\begin{array}{lll}
w \models_{M} P(t_{1}, \ldots, t_{m}) & \iff& \langle V(t_{1}), \ldots, V(t_{m}) \rangle \in V(P,w),\\
w \not\models_{M} \bot & \text{}\\
w \models_{M} A \land B & \iff & w \models_{M} A \text{ and } w \models_{M} B,\\
w \models_{M} A \lor B & \iff & w \models_{M} A \text{ or } w \models_{M} B,\\
w \models_{M} A \imp_{\mathtt{i}} B & \iff & \text{for all } v \in W, (wRv \text{ and } v \models_{M} A \text{ jointly imply } v \models_{M} B),\\
w \models_{M} A \imp_{\mathtt{c}} B & \iff & w \models_{M} A \text{ implies } w \models_{M} B,\\
w \models_{M} \forall_{\mathtt{i}}x A & \iff & \text{ for all } v \in W, (wRv \text{ implies } \text{for all } d \in D(v), v \models_{M} A[\underline{d} / x]),
\\
w \models_{M} \forall_{\mathtt{c}}x A & \iff & \text{ for all } d \in D(w), w \models_{M} A[\underline{d} / x],\\
w \models_{M} \exists x A & \iff & \text{ for some } d \in D(w), w \models_{M} A[\underline{d} / x].\\
\end{array}
\]
\end{definition}

\noindent We say that a closed formula $A$ satisfies {\em heredity} if, for every model $M$ = $(W,R,(D(w))_{w \in W},V)$ and every $w, v \in W$, $w \models_{M} A$ and $w R v$ jointly imply $v \models_{M} A$. The notion of semantic consequence relation is defined as below.

\begin{definition}
\label{def:sec}
Suppose $\Gamma \cup \{ A \} \subseteq \mathsf{Form}$. Then, the {\em semantic consequence} relation $\Gamma \models A$ is defined as follows:
For all $M$ = $(W,R,(D(w))_{w \in W},V)$, all $w \in W$ and all $d:\mathsf{Var} \to D(w)$, if\linebreak $w \models_{M} C  [\underline{d (x_{1})}/x_{1}]\cdots [\underline{d (x_{m})}/x_{m}]$ for any $C \in \Gamma$, then $w \models A[\underline{d (z_{1})}/z_{1}] \cdots [\underline{d (z_{n})}/z_{n}]$, where $x_{1},\ldots, x_{m}$ are all free variables in $C$ and $z_{1}, \ldots, z_{n}$ are all free variables in $A$. 
A formula $A$ is {\em valid} if $\models A$ holds.
\end{definition}

\noindent By Definition \ref{def:sec}, we can consider semantic consequence relation not only of closed but also of possibly non-closed formulas.

We proceed to the matter of heredity in a Kripke model. Heredity is known as an important feature of pure intuitionistic logic, i.e., Fact \ref{fact:her} holds.

\begin{fact}
\label{fact:her}
All closed formulas $A \in \mathsf{Form}_{\mathbf{J}}$ i.e., the set of formulas expressive in the syntax of the intuitionistic logic, satisfy heredity.
\end{fact}

\noindent However, this feature is lost when we add classical implication and universal quantifier to the intuitionistic logic. Consider a predicate $P$ whose arity is one, and consider a model $M$ = $(W,R,(D(w))_{w \in W},V)$ such that $W=\{w ,v \},R=\{(w, w),(w, v), (v,v)\}, D(w) = D(v) = \{ d \}$, $V(c)$ = $d \in D(w) \cap D(v)$, $ d  \notin V(P,w)$, and $d \in V(P,v)$. In this model, $wRv$ and $w \models_{M} \neg_{\mathtt{c}} P(c)$ hold, but $v \not\models_{M} \neg_{\mathtt{c}} P(c)$ holds. Moreover, $w \not\models_{M} \top \imp_{\mathtt{c}}  \neg_{\mathtt{c}} P(c)$ and $w \not\models_{M} \top \imp_{\mathtt{i}}  \neg_{\mathtt{c}} P(c)$. These arguments give us the following propositions.

\begin{prop}
\label{prop:per1}
A formula $\neg_{\mathtt{c}} P(c)$ does not satisfy heredity. 
\end{prop}

\begin{prop}
\label{prop:per2}
Neither $\neg_{\mathtt{c}} P(c) \imp_{\mathtt{i}} (\top \imp_{\mathtt{i}}  \neg_{\mathtt{c}} P(c))$ nor $\neg_{\mathtt{c}} P(c) \imp_{\mathtt{c}} (\top \imp_{\mathtt{i}} \neg_{\mathtt{c}} P(c))$ is valid. 
\end{prop}

\noindent By Proposition \ref{prop:per2}, it is obvious that an intuitionistic tautology $A \imp_{\mathtt{i}} (B \imp_{\mathtt{i}} A)$ is no longer valid. As we have mentioned in Section \ref{sec:int}, this fact affects the construction of the Hilbert system in~\cite{Cerro1996}. 
A similar phenomenon also happens about classical universal quantifier. Consider the same model $M$ as above. Then, $wRv$ and $w \models_{M} \neg_{\mathtt{c}} \forall_{\mathtt{c}}x P(x)$ hold, but $v \not\models_{M} \neg_{\mathtt{c}} \forall_{\mathtt{c}}x P(x)$. Therefore, we get the following proposition.

\begin{prop}
\label{prop:nhc}
A formula $\neg_{\mathtt{c}} \forall_{\mathtt{c}}x P(x)$ does not satisfy heredity.
\end{prop}

\noindent Let us still consider the same model $M$. Then, $w \models_{M} \forall_{\mathtt{i}}y (\neg_{\mathtt{c}} \forall_{\mathtt{c}}x P(x) \imp_{\mathtt{i}} \neg_{\mathtt{c}} P(y))$ and $wRw$ hold. But $w \not\models_{M} \neg_{\mathtt{c}} \forall_{\mathtt{c}}x P(x) \imp_{\mathtt{i}} \forall_{\mathtt{i}}y \neg_{\mathtt{c}} P(y)$ holds. This is because $w \models_{M} \neg_{\mathtt{c}} \forall_{\mathtt{c}}x P(x)$ and $wRw$, but $w \not\models \forall_{\mathtt{i}}y \neg_{\mathtt{c}} P(y)$. This gives us the following proposition.

\begin{prop}
\label{prop:bre}
A formula $\forall_{\mathtt{i}}y (\neg_{\mathtt{c}} \forall_{\mathtt{c}}x P(x) \imp_{\mathtt{i}} \neg_{\mathtt{c}} P(y)) \imp_{\mathtt{i}} (\neg_{\mathtt{c}} \forall_{\mathtt{c}}x P(x) \imp_{\mathtt{i}} \forall_{\mathtt{i}}y  \neg_{\mathtt{c}} P(y))$ is not valid. 
\end{prop}

\noindent Proposition \ref{prop:bre} implies an intuitionistic tautology $\forall_{\mathtt{i}}y(A \imp_{\mathtt{i}} B) \imp_{\mathtt{i}} (A \imp_{\mathtt{i}} \forall_{\mathtt{i}}y B)$, where $y$ does not occur free in $A$, is no longer valid, either.

\subsection{Sequent Calculus $\fljc$}
This section provides a sequent calculus $\fljc$, which is a first-order expansion of propositional $\ljc$ in~\cite{Toyooka2021a}. The calculus employs the ordinary notion of multi-succedent sequent. A {\em sequent} is a pair of finite multisets of formulas denoted by 
$\Gamma \Rightarrow \Delta$, which is read as ``if all formulas in $\Gamma$ hold then some formulas in $\Delta$ hold.'' 

The sequent calculus $\fljc$ consists of the axioms and the rules in Table \ref{fig:rules}. The notion of derivability is defined as an existence of a finite tree, which is called a {\em derivation}, generated by inference rules of Table \ref{fig:rules} from initial sequents $(Id)$ and $(\bot)$ of Table \ref{fig:rules}.

Our basic strategy of constructing $\fljc$ is to add classical implication and universal quantifier to the multi-succedent sequent calculus $\mathbf{mLJ}$, proposed by Maehara~\cite{Maehara1954}, 
where the right rules for intuitionistic implication and universal quantifier are of the following form:
\[
\infer
    {\Gamma \Rightarrow  A \imp_{\mathtt{i}} B}
    {A, \Gamma \Rightarrow B},
\qquad
\infer{\Gamma \Rightarrow \forall_{\mathtt{i}}x A}{\Gamma \Rightarrow A [z/x]},
\]
where $\Gamma$ is a finite multiset of intuitionistic formulas.
However, if the ordinary left and right rules for classical implication were added to $\mathbf{mLJ}$, the soundness of the resulting calculus would fail, because formulas $\neg_{\mathtt{c}} P(t) \imp_{\mathtt{c}} (\top \imp_{\mathtt{i}}  \neg_{\mathtt{c}} P(t))$ and $\forall_{\mathtt{i}}y (\neg_{\mathtt{c}} \forall_{\mathtt{c}}x P(x) \imp_{\mathtt{i}} \neg_{\mathtt{c}} P(y)) \imp_{\mathtt{i}} (\neg_{\mathtt{c}} \forall_{\mathtt{c}}x P(x) \imp_{\mathtt{i}} \forall_{\mathtt{i}}y  \neg_{\mathtt{c}} P(y))$ would be derivable, which are found invalid by Propositions \ref{prop:per2} and \ref{prop:bre}. This is the reason why the original right rules for intuitionistic implication and universal quantifier of $\mathbf{mLJ}$ described above are restricted to the right rules given in Table \ref{fig:rules}. Although the restriction should be imposed on a context $\Gamma$ of theses rules, no extra restriction is needed. Thus, the following applications of $(\Rightarrow \imp_{\mathtt{i}})$ and $(\Rightarrow \forall_{\mathtt{i}})$ are always legitimate:

\[
\infer[(\Rightarrow \imp_{\mathtt{i}})]{\Rightarrow A \imp_{\mathtt{i}} B}{A \Rightarrow B},
\qquad
\infer[(\Rightarrow \forall_{\mathtt{i}})]{\Rightarrow \forall_{\mathtt{i}}x A}{\Rightarrow A [z/x]}.
\]

Based on the abbreviation defined in Section \ref{subsec:syntax}, the following rules for the negations are obtained.

\[
\infer[(\Rightarrow \neg_{\mathtt{i}})]{\forall_{\mathtt{i}}x_{1} B_{1}, \ldots, \forall_{\mathtt{i}}x_{l} B_{l}, C_{1} \imp_{\mathtt{i}} D_{1}, \ldots, C_{m} \imp_{\mathtt{i}} D_{m}, p_{1}, \ldots p_{n} \Rightarrow  \neg_{\mathtt{i}} A}{A, \forall_{\mathtt{i}}x_{1} B_{1}, \ldots, \forall_{\mathtt{i}}x_{l} B_{l}, C_{1} \imp_{\mathtt{i}} D_{1}, \ldots, C_{m} \imp_{\mathtt{i}} D_{m}, p_{1}, \ldots p_{n} \Rightarrow}
\quad
\infer[(\neg_{\mathtt{i}} \Rightarrow)]
    {\neg_{\mathtt{i}} A, \Gamma \Rightarrow  \Delta}
    {\Gamma \Rightarrow \Delta, A}
\]
\[
\infer[(\Rightarrow \neg_{\mathtt{c}})]{\Gamma \Rightarrow \neg_{\mathtt{c}} A, \Delta}{\Gamma, A \Rightarrow \Delta}
\quad
\infer[(\neg_{\mathtt{c}} \Rightarrow)]{\neg_{\mathtt{c}} A, \Gamma \Rightarrow \Delta}{\Gamma \Rightarrow \Delta, A}.
\]

For example, the sequent $\top \imp_{\mathtt{i}} \forall_{\mathtt{c}}x P(x) \Rightarrow \forall_{\mathtt{c}}x (\top \imp_{\mathtt{i}} P(x))$ is derivable in $\fljc$ as follows:

\[
\infer[(\Rightarrow \forall_{\mathtt{c}})]{\top \imp_{\mathtt{i}} \forall_{\mathtt{c}}x P(x) \Rightarrow \forall_{\mathtt{c}}x (\top \imp_{\mathtt{i}} P(x))}{\infer[(\Rightarrow \imp_{\mathtt{i}})]{\top \imp_{\mathtt{i}} \forall_{\mathtt{c}}x P(x) \Rightarrow \top \imp_{\mathtt{i}} P(z)}{\infer[(\imp_{\mathtt{i}} \Rightarrow)]{\top \imp_{\mathtt{i}} \forall_{\mathtt{c}}x P(x), \top \Rightarrow P(z)}{{\top \Rightarrow \top} & {\infer[(\forall_{\mathtt{c}} \Rightarrow)]{\forall_{\mathtt{c}}x P(x) \Rightarrow P(z)}{P(z) \Rightarrow P(z)}}}}}.
\]

\begin{table}[htbp]
    \caption{Sequent Calculus $\fljc$}
    \label{fig:rules}
\hrule
\textbf{Axioms}
\[
\infer[(Id)]{A \Rightarrow A}{}
~~~ \infer[(\bot)]{\bot \Rightarrow}{}
\]
\textbf{Structural Rules}
\[
\infer[(\Rightarrow w)]{\Gamma \Rightarrow  \Delta, A}{\Gamma \Rightarrow \Delta}~~~
\infer[(w \Rightarrow)]{A , \Gamma \Rightarrow \Delta}{\Gamma \Rightarrow \Delta}~~~
 \infer[(\Rightarrow c)]{\Gamma \Rightarrow \Delta, A}{\Gamma \Rightarrow \Delta,A,A}~~~
\infer[(c \Rightarrow)]{A , \Gamma \Rightarrow \Delta}{A,A, \Gamma \Rightarrow \Delta}
\]
\[
\infer[(Cut)]
    {\Gamma, \Pi \Rightarrow \Delta, \Sigma}
    {\Gamma \Rightarrow \Delta, A
    &
    A, \Pi \Rightarrow \Sigma}
\]
\textbf{Logical Rules}
\[
\infer[(\Rightarrow \imp_{\mathtt{i}})]{\forall_{\mathtt{i}}x_{1} C_{1}, \ldots, \forall_{\mathtt{i}}x_{l} C_{l}, D_{1} \imp_{\mathtt{i}} E_{1}, \ldots, D_{m} \imp_{\mathtt{i}} E_{m}, p_{1}, \ldots p_{n} \Rightarrow A \imp_{\mathtt{i}} B}{A, \forall_{\mathtt{i}}x_{1} C_{1}, \ldots, \forall_{\mathtt{i}}x_{l} C_{l}, D_{1} \imp_{\mathtt{i}} E_{1}, \ldots, D_{m} \imp_{\mathtt{i}} E_{m}, p_{1}, \ldots p_{n} \Rightarrow B}
\]
\[
\infer[(\imp_{\mathtt{i}} \Rightarrow)]
    {A \imp_{\mathtt{i}} B, \Gamma_{1}, \Gamma_{2} \Rightarrow  \Delta_{1}, \Delta_{2}}
    {\Gamma_{1} \Rightarrow \Delta_{1}, A
    &
    B, \Gamma_{2} \Rightarrow \Delta_{2}}
\]
\[
\infer[(\Rightarrow \imp_{\mathtt{c}})]{\Gamma \Rightarrow \Delta, A \imp_{\mathtt{c}} B}{A, \Gamma \Rightarrow \Delta, B}~~~
\infer[(\imp_{\mathtt{c}} \Rightarrow)]
    {A \imp_{\mathtt{c}} B, \Gamma_{1}, \Gamma_{2} \Rightarrow  \Delta_{1}, \Delta_{2}}
    {\Gamma_{1} \Rightarrow \Delta_{1}, A
    &
    B, \Gamma_{2} \Rightarrow \Delta_{2}}
\]
\[ \infer[(\Rightarrow\land)]
    {\Gamma \Rightarrow \Delta, A \land B}
    {\Gamma \Rightarrow \Delta, A
    &
    \Gamma \Rightarrow \Delta, B}~~~
\infer[(\land\Rightarrow_1)]
    {A\land B, \Gamma \Rightarrow \Delta}
    {A, \Gamma  \Rightarrow \Delta}~~~
  \infer[(\land \Rightarrow_2)]
    {A\land B, \Gamma \Rightarrow \Delta}
    {B, \Gamma  \Rightarrow \Delta}
\]
\[
\infer[(\Rightarrow\lor_1)]{\Gamma \Rightarrow \Delta, A \lor B}{\Gamma \Rightarrow \Delta, A}~~~
\infer[(\Rightarrow\lor_2)]{\Gamma \Rightarrow \Delta, A \lor B}{\Gamma \Rightarrow \Delta, B}~~~
\infer[(\lor\Rightarrow)]
    {A \lor B, \Gamma \Rightarrow \Delta}
    {A, \Gamma \Rightarrow \Delta
    &
    B, \Gamma \Rightarrow \Delta}
\]
\[
\infer[(\Rightarrow \forall_{\mathtt{i}})^{\dag}]{\forall_{\mathtt{i}}x_{1} B_{1}, \ldots, \forall_{\mathtt{i}}x_{l} B_{l}, C_{1} \imp_{\mathtt{i}} D_{1}, \ldots, C_{m} \imp_{\mathtt{i}} D_{m}, p_{1}, \ldots p_{n} \Rightarrow \forall_{\mathtt{i}}x A}{\forall_{\mathtt{i}}x_{1} B_{1}, \ldots, \forall_{\mathtt{i}}x_{l} B_{l}, C_{1} \imp_{\mathtt{i}} D_{1}, \ldots, C_{m} \imp_{\mathtt{i}} D_{m}, p_{1}, \ldots p_{n} \Rightarrow A[z/x]}~~~
\infer[(\forall_{\mathtt{i}} \Rightarrow)]{\forall_{\mathtt{i}}x A, \Gamma \Rightarrow \Delta}{A[t/x], \Gamma \Rightarrow \Delta}
\]
\[
\infer[(\Rightarrow \forall_{\mathtt{c}})^{\dag}]{\Gamma \Rightarrow \Delta, \forall_{\mathtt{c}}x A}{\Gamma \Rightarrow \Delta, A[z/x]}~~~
\infer[(\forall_{\mathtt{c}} \Rightarrow)]{\forall_{\mathtt{c}}x A, \Gamma \Rightarrow \Delta}{A[t/x], \Gamma \Rightarrow \Delta}
\]
\[
\infer[(\Rightarrow \exists)]{\Gamma \Rightarrow \Delta, \exists x A}{\Gamma \Rightarrow \Delta, A[t/x]}~~~
\infer[(\exists \Rightarrow)^{\dag}]{\exists x A, \Gamma \Rightarrow \Delta}{A[z/x], \Gamma \Rightarrow \Delta}
\]
\begin{center}
    $\dag$: $z$ does not occur free in the lower sequent.
\end{center}
\hrule
\end{table}

As is noted in Section \ref{sec:int} for the propositional $\ljc$, although the context of the right rule for intuitionistic implication is restricted, all sequents derivable in propositional intuitionistic logic are also derivable. Similarly, all sequents derivable in first-order $\mathbf{mLJ}$ are also derivable in $\fljc$. The following proposition ensures this.

\begin{prop}
The ordinary right rules for intuitionistic implication and universal quantifier are derivable in $\fljc$, if the contexts of the rules contain only formulas in $\mathsf{Form}_{\mathbf{J}}$.
\end{prop}

\noindent The ordinary right rules for intuitionistic implication and universal quantifier are derivable also in the calculus obtained from $\mathbf{mLJ}$ by replacing them with the restricted ones. This means the restricted version of the rules denotes the core of the ordinary ones.

We proceed to the soundness theorem. We use $\Gamma \models \Delta$ to mean: for some formula $C \in \Delta$, $\Gamma \models C$ holds. 

\begin{definition}
We define {\em persistent} formulas inductively as follows:
\[
E ::= \, P(t_{1}, \ldots, t_{m}) \,|\, A \imp_{\mathtt{i}} A \,|\, \forall_{\mathtt{i}}x A,
\]
where $P(t_{1}, \ldots, t_{m}) \in \mathsf{Form}$ and $A  \in \mathsf{Form}$.
\end{definition}
Formulas occurring at a context of the right rules for the intuitionistic implication and universal quantifier must be persistent. In what follows, we use $\Theta$ to denote a multiset of persistent formulas. 
This is just for making the notation simpler. By this notation, the right rules for the intuitionistic implication and the intuitionistic universal quantifier are described as follows:

\[
\infer[(\Rightarrow \imp_{\mathtt{i}})]{\Theta \Rightarrow A \imp_{\mathtt{i}} B}{A, \Theta \Rightarrow B},
\quad
\infer[(\Rightarrow \forall_{\mathtt{i}})]{\Theta \Rightarrow \forall_{i}x A}{\Theta \Rightarrow A[z/x]}.
\]

\noindent Persistent formulas satisfy heredity, which is trivial from Fact \ref{fact:her}.

\begin{thm}
\label{thm:fcs}
If $\Gamma \Rightarrow \Delta$ is derivable in $\fljc$, then $\Gamma \models \Delta$ holds.  
\end{thm}

\begin{proof}
It can be shown straightforwardly that every axiom and rule in $\fljc$ except for the rules $(\Rightarrow \imp_{\mathtt{i}})$ and $(\Rightarrow \forall_{\mathtt{i}})$ preserves validity. 
Only the cases of $(\Rightarrow \imp_{\mathtt{i}})$ and $(\Rightarrow \forall_{\mathtt{i}})$ are considered here. 
\begin{description}
    \item[$(\Rightarrow \imp_{\mathtt{i}})$] Let $x_{1}, \ldots, x_{m}$ be free variables of a formula in $\Theta$ and $z_{1}, \ldots, z_{n}$ be free variables in\linebreak $A \imp_{\mathtt{i}} B$. Suppose $w \models_{M} \left(\bigwedge \Theta \right)[\underline{d(x_{1})}/x_{1}] \cdots [\underline{d(x_{m})}/x_{m}]$. 
We have to show\linebreak $w \models_{M} (A \imp_{\mathtt{i}} B)[\underline{d(z_{1})}/z_{1}] \cdots [\underline{d(z_{n})}/z_{n}]$, i.e., for all $v \in W, (wRv \text{ and }\linebreak v \models_{M} A[\underline{d(z_{1})}/z_{1}] \cdots [\underline{d(z_{n})}/z_{n}]$ jointly imply $v \models_{M} B[\underline{d(z_{1})}/z_{1}] \cdots [\underline{d(z_{n})}/z_{n}])$. Fix any $v$ which satisfies $wRv$ and $v \models_{M} A[\underline{d(z_{1})}/z_{1}] \cdots [\underline{d(z_{n})}/z_{n}]$. 
Then, since all of the formulas in $\Theta$ satisfy heredity, we obtain $v \models_{M} \left(\bigwedge \Theta \right)[\underline{d(x_{1})}/x_{1}] \cdots [\underline{d(x_{m})}/x_{m}]$. 
By the validity of the premise of $(\Rightarrow \imp_{\mathtt{i}})$, we obtain $v \models_{M} B[\underline{d(z_{1})}/z_{1}] \cdots [\underline{d(z_{n})}/z_{n}]$, as required. 

\item[$(\Rightarrow \forall_{\mathtt{i}})$] Let $x_{1}, \ldots, x_{m}$ be free variables of a formula in $\Theta$ and $z_{1}, \ldots, z_{n}$ be free variables in $\forall_{\mathtt{i}}x A$. 
Suppose $w \models_{M} \left(\bigwedge \Theta \right)[\underline{d(x_{1})}/x_{1}] \cdots [\underline{d(x_{m})}/x_{m}]$.
We have to show $w \models_{M} (\forall_{\mathtt{i}}x A)[\underline{d(z_{1})}/z_{1}] \cdots [\underline{d(z_{n})}/z_{n}]$, i.e., for all $v \in W$, ($wRv$ implies for all $d \in D(v)$, $v \models_{M} A[\underline{d(z_{1})}/z_{1}] \cdots [\underline{d(z_{n})}/z_{n}][\underline{d}/x]$). Fix any $v$ which satisfies $wRv$ and any $d \in D(v)$. Consider an assignment $d(z|d)$ satisfying the following: $d(z|d)(y)$ = $d(y)$ if $y \not\equiv z$ and $d(z|d)(y)$ = $d$ if $y \equiv z$. Since $D(v)$ contains all the elements of $D(w)$, the assignment $d(z|d)$ can be regarded as an assignment: $\mathsf{Var} \to D(v)$. By the validity of the premise of $(\Rightarrow \forall_{\mathtt{i}})$, we obtain the following: if $v \models_{M} \left(\bigwedge \Theta \right)[\underline{d(z|d)(x_{1})}/x_{1}] \cdots [\underline{d(z|d)(x_{m})}/x_{m}]$ hold, then $v \models_{M} A[z/x][\underline{d(z|d)(z_{1})}/z_{1}] \cdots [\underline{d(z|d)(z_{n})}/z_{n}]$ hold. Then, since all formulas in $\Theta$ satisfy heredity, $v \models_{M} \left(\bigwedge \Theta \right)[\underline{d(x_{1})}/x_{1}] \cdots [\underline{d(x_{m})}/x_{m}]$ holds. Since $z$ does not occur free in any formula $E \in \Theta$, $v \models_{M} A[z/x][\underline{d(z|d)(z_{1})}/z_{1}] \cdots [\underline{d(z|d)(z_{n})}/z_{n}]$ is obtained. By this, we can obtain $v \models_{M} A[\underline{d(z_{1})}/z_{1}] \cdots [\underline{d(z_{n})}/z_{n}][\underline{d}/x]$, as required.
\qedhere 
\end{description}
\end{proof}

\section{Cut Elimination}
\label{sec:cut}
In this section, it is supposed that the sets of bound and free variables are disjoint, as in~\cite{Kashima2009}. Let us denote by $\fljcc$ a sequent calculus obtained from $\fljc$ by removing the rule $(Cut)$. In~\cite{Toyooka2021a}, with the help of a variant of ``Mix rule'' by Gentzen (``{\em extended cut rule}'' used in~\cite{Kashima2009,Ono2019,Ono1985}) to take care of contraction rules, cut elimination theorem was already shown for propositional logic calculus $\ljc$. ``{\em Extended cut rule}''$(Ecut)$ is described as follows:

\[
\infer[(Ecut)]{\Gamma, \Pi \Rightarrow \Delta, \Sigma}{\Gamma \Rightarrow \Delta, A^{m} & A^{n}, \Pi \Rightarrow \Sigma},
\]
where $m,n \geq 0$ and $A^{l}$ means $\underbrace{A, \ldots, A}_{l \text{ times occurrences}}$, i.e., $l$ times repetition of $A$. 

In order to show the cut elimination theorem, it suffices to deal with a derivation of the ``{\em $(Ecut)$-bottom form}'' defined in Definition \ref{def:ecr}.

 \begin{definition}
\label{def:ecr}
A derivation $\mathcal{D}$ in $\fljc$ is of the {\em $(Ecut)$-bottom form} if it has the following form
\[
\infer[(Ecut)]{\Gamma, \Pi \Rightarrow \Delta, \Sigma}{{\infer[rule (\mathcal{D}_{1})]{\Gamma \Rightarrow \Delta, A^{m}}{\vdots & \mathcal{D}_{1}}}&{\infer[rule(\mathcal{D}_{2})]{A^{n}, \Pi \Rightarrow \Sigma}{\vdots & \mathcal{D}_{2}}}},
\]
where $rule(\mathcal{D}_{i})$ is the last applied rule of a given derivation $\mathcal{D}_{i}$, and there is no application of $(Ecut)$ in $\mathcal{D}_{1}$ nor $\mathcal{D}_{2}$. Let the {\em weight} $w$ of an Ecut-bottom form be the sum of the number of sequents occurring in $\mathcal{D}_{1}$ and $\mathcal{D}_{2}$. Let the {\em complexity} $c$ of an Ecut-bottom form be the number of logical symbols ($\imp_{\mathtt{i}}, \imp_{\mathtt{c}}, \land, \lor, \forall_{\mathtt{i}}, \forall_{\mathtt{c}}$, and $\exists$) appearing in the {\em Ecut formula}. It should be noted that this means substitution does not change the complexity.  For example, $A$ and $A[t/x]$ has the same complexity. We note that $w \geq 2$ and $c \geq 0$.
\end{definition}

\begin{definition}
\label{def:pri}
We define a {\em principal formula} of a logical rule as follows.
\begin{itemize}
    \item For every logical rule except for $(\Rightarrow \imp_{\mathtt{i}})$ and $(\Rightarrow \forall_{\mathtt{i}})$, a {\em principal formula} of the rule is the unique compound formula in the lower sequent of the rule which is produced by an application of the logical rule. 
    \item For $(\Rightarrow \imp_{\mathtt{i}})$ and $(\Rightarrow \forall_{\mathtt{i}})$, every formula occurring in the lower sequent of the rules is {\em principal}. 
\end{itemize}
\end{definition}

Based on Definitions \ref{def:ecr} and \ref{def:pri}, we can show Lemma \ref{lem:cut}, which is the core of showing the cut elimination theorem. In order to show this lemma, Lemma \ref{lem:subst}, which is related to substitution, is needed. Recall that $[t/x]$ is {\em clash avoiding substitution} of $t$ for $x$.

\begin{lem}
\label{lem:subst}
If there is a derivation $\mathcal{D}$ of $\Gamma \Rightarrow \Delta$, then there is also a derivation $\mathcal{D}'$ of  $\Gamma[t/x] \Rightarrow \Delta[t/x]$ whose weight is the same as that of $\mathcal{D}$.
\end{lem}

\noindent Lemma \ref{lem:subst} can be shown by induction on the weight of a derivation, as was done in~\cite{Negri2001,Troelstra2012}. 

\begin{lem}
\label{lem:cut}
For every derivation of the {$(Ecut)$-bottom form}, there is an {$(Ecut)$-free derivation} with the same conclusion.
\end{lem}

\begin{proof}
 By double induction on the complexity and the weight lexicographically. For the case of $m = 0$ or $n = 0$, $(Ecut)$ can be eliminated by applying $(\Rightarrow w)$ or $(w \Rightarrow)$. Thus, we assume $m > 0$ and $n > 0$ in what follows. For the other cases, our argument is divided into the following four cases:

\begin{enumerate}
    \item $\mathcal{D}_{i}$ is an initial sequent,
    \item $rule(\mathcal{D}_{i})$ is a structural rule,
    \item $rule(\mathcal{D}_{i})$ is a logical rule where the {\em Ecut formula} is not principal,
    \item $rule(\mathcal{D}_{1})$ and $rule(\mathcal{D}_{2})$ are logical rules, and the {\em Ecut formulas} are principal in both rules.
\end{enumerate}

\noindent The definition of a principal formula (Definition \ref{def:pri}) enables cases categorized into $(3)$ to be treated with little difficulty. Here, we only shows the case when $rule(\mathcal{D}_{1})$ is $(\Rightarrow \imp_{\mathtt{i}})$ and $rule(\mathcal{D}_{2})$ is $(\Rightarrow \forall_{\mathtt{i}})$, which is categorized into $(4)$. In this case, a derivation of the $(Ecut)$-bottom form is described as follows:
\[
\infer[(Ecut)]{\Theta_{1}, \Theta_{2} \Rightarrow \forall_{\mathtt{i}} x A}{{\infer[(\Rightarrow \imp_{\mathtt{i}})]{\Theta_{1} \Rightarrow B \imp_{\mathtt{i}} C}{\infer{B, \Theta_{1} \Rightarrow C}{\vdots &  \mathcal{D}_{1}}}} & {\infer[(\Rightarrow \forall_{\mathtt{i}})]{(B \imp_{\mathtt{i}} C)^{n}, \Theta_{2} \Rightarrow \forall_{\mathtt{i}}x A}{\infer{(B \imp_{\mathtt{i}} C)^{n}, \Theta_{2} \Rightarrow A[z/x]}{\vdots &  \mathcal{D}_{2}}}}},
\]
where $\Theta_{1}$ and $\Theta_{2}$ are multisets of persistent formulas, and $z$ does not occur free in the sequent $(B \imp_{\mathtt{i}} C)^{n}, \Theta_{2} \Rightarrow \forall_{\mathtt{i}}x A$. Since the sets of bound and free variables are disjoint, $z$ does not occur in $(B \imp_{\mathtt{i}} C)^{n}, \Theta_{2} \Rightarrow \forall_{\mathtt{i}}x A$. However, $z$ can occur in $\Theta_{1}$. Suppose $y$ does not occur free in $\Theta_{1}$, $B \imp_{\mathtt{i}} C$, $\Theta_{2}$, or $\forall_{\mathtt{i}}x A$, which implies $y$ does not occur in these formulas. Then, by Lemma \ref{lem:subst}, we can obtain the sequent $((B \imp_{\mathtt{i}} C)[y/z])^{n}, \Theta_{2}[y/z] \Rightarrow A[z/x][y/z]$. Since $z$ does not occur in $B \imp_{\mathtt{i}} C$, $\Theta_{2}$, or $\forall_{\mathtt{i}}x A$, a sequent $(B \imp_{\mathtt{i}} C)^{n}, \Theta_{2} \Rightarrow A[y/x]$ is obtained. By using this sequent, we can obtain the following derivation:
\[
\infer[(\Rightarrow \forall_{\mathtt{i}})]{\Theta_{1}, \Theta_{2} \Rightarrow \forall_{\mathtt{i}}x A}{\infer[(Ecut)]{\Theta_{1}, \Theta_{2} \Rightarrow A[y/x]}{{\infer{\Theta_{1} \Rightarrow B \imp_{\mathtt{i}} C}{\vdots & \mathcal{D}_{1}}} & {\infer{(B \imp_{\mathtt{i}} C)^{n}, \Theta_{2} \Rightarrow A[y/x]}{\vdots & \mathcal{D}'_{2}}}}}.
\]

\noindent Since the weight of the derivation of the $(Ecut)$-bottom form becomes lesser, we can apply induction hypothesis and obtain an  $(Ecut)$-free derivation.
\end{proof}

Finally, the cut elimination theorem is obtained, as required.
\begin{thm}
\label{thm:cet}
If $\Gamma \Rightarrow \Delta$ is derivable in $\fljc$, then $\Gamma \Rightarrow \Delta$ is derivable in $\fljcc$.
\end{thm}

\noindent By Theorem \ref{thm:cet}, the subformula property is also obtained, which ensures the following corollary. 

\begin{corollary}
\label{cor:cet}
The sequent calculus $\fljc$ is a conservative extension of both first-order intuitionistic and classical logic.
\end{corollary}

\section{Strong Completeness}
\label{sec:comp}
This section establishes the strong completeness theorem of $\fljc$. In~\cite{Cerro1996,Humberstone1979}, the completeness of propositional $\mathbf{C+J}$ was shown. In~\cite{Toyooka2021a}, the completeness of $\ljc$ was shown by establishing the fact that the calculus and the Hilbert system of $\mathbf{C+J}$ proposed in~\cite{Cerro1996} are {\em equipollent} in the following sense:

\begin{quote}
For any formula $A \in \mathsf{Form}_{\mathbf{C+J}}$, the sequent  $\Rightarrow A$ is derivable in $\ljc$ iff $A$ is derivable in the Hilbert system of $\mathbf{C+J}$.
\end{quote}

\noindent Only the weak completeness was shown in~\cite{Cerro1996,Humberstone1979,Toyooka2021a}, but the strong completeness of propositional $\ljc$ can be obtained by reinforcing a canonical model argument described in~\cite{Humberstone1979}. The strong completeness of the first-order combination $\fljc$, which will be established in this section via a canonical model argument, has not been shown so far. 

In this section, the expressions such as $\Gamma,\Delta,\Theta$ denote sets (not multisets) of formulas. In order to deal with a sequent which contains a possibly infinite set of formulas, we have to expand the notion of derivability as follows:

\begin{quote}
    $\Gamma \Rightarrow \Delta$ is derivable in $\fljc$ if for some finite subset $\Gamma'$ of $\Gamma$ and $\Delta'$ of $\Delta$, $\Gamma' \Rightarrow \Delta'$ is derivable in $\fljc$.
\end{quote}

\noindent We also have to expand the syntax $\mathcal{L}$ to $\mathcal{L^{+}}$ by adding a new countably infinite set of variables.

\begin{definition}
Let $L_{1}$ and $L_{2}$ be any syntax which has the same logical symbols and contains all constant and predicate symbols of $\mathcal{L}$. Recall that $\mathsf{Var}(L_{i})$ is the set of variables of $L_{i}$. We use $L_{1} \sqsubset L_{2}$ to mean 
\begin{center}
    $\mathsf{Var}(L_{1}) \subsetneq \mathsf{Var}(L_{2})$ and $\#(\mathsf{Var}(L_{2}) \setminus \mathsf{Var}(L_{1}))$ = $\omega$.
\end{center}
\end{definition}

Then, we should define a prime pair of sets of formulas with respect to a syntax as did in~\cite{Gabbay2009} to show the completeness of an ordinary first-order intuitionistic logic.

\begin{definition}
\label{def:prime}
A pair ${( \Gamma, \Delta )}_{L}$ of sets of formulas with respect to a  syntax $L$ is a {\em prime pair} with respect  to $L$ if it satisfies the following:
\begin{description}
    \item[($\Gamma$ is a theory)] If $\Gamma \Rightarrow A$ is derivable in $\fljc$, then $A \in \Gamma$,
    \item[(underivability)] $\Gamma \Rightarrow \Delta$ is not derivable in $\fljc$,
    \item[(primeness)] If $A \lor B \in \Gamma$ holds, then $A \in \Gamma$ or $B \in \Gamma$ holds,
    \item[($\exists$-property)] If $\exists x A \in \Gamma$ holds, then for some term $t$ of $L$, $A[t/x] \in \Gamma$ holds,
    \item[($\forall_{\mathtt{c}}$-property)] If $\forall_{\mathtt{c}}x A \in \Delta$ holds, then for some term $t$ of $L$, $A[t/x] \in \Delta$ holds.
\end{description}
A prime pair ${( \Gamma, \Delta )}_{L}$ with respect to $L$ is called {\em $L$-complete} if $\Gamma \cup \Delta$ = $\mathsf{Form}_{L}$ holds.
\end{definition}

\noindent The consistency of a prime pair $(\Gamma, \Delta)_{L}$, i.e., $\bot \notin \Gamma$ can be obtained from the condition (underivability).

\begin{lem}
\label{lem:ccc}
A prime pair ${( \Gamma, \Delta )}_{L}$ with respect to a syntax $L$ is {\em classical negation complete}, i.e., for all formulas $A$ in $L$, either $A \in \Gamma$ or $A \imp_{\mathtt{c}} \bot \in \Gamma$ holds.
\end{lem}

\begin{proof}
Let ${( \Gamma, \Delta )}_{L}$ be a prime pair with respect to a syntax $L$ and $A$ be a formula in $L$. In $\fljc$, $\Rightarrow A \lor (A \imp_{\mathtt{c}} \bot)$ is derivable. Then, by the condition ($\Gamma$ is a theory) in Definition \ref{def:prime}, it follows  $A \lor (A \imp_{\mathtt{c}} \bot) \in \Gamma$. By (primeness) in Definition \ref{def:prime}, either $A \in \Gamma$ or $A \imp_{\mathtt{c}} \bot \in \Gamma$, as is desired. 
\end{proof}

The following lemma implies that any underivable pair of sets of formulas can be extended to a prime $L$-complete pair 
for some appropriate syntax $L$. 

\begin{lem}
\label{lem:com}
\begin{enumerate}
    \item If a sequent $\Gamma \Rightarrow \Delta$ in $\mathcal{L}$ is not derivable in $\fljc$, then there exist a syntax $L$ and a prime $L$-complete pair ${( \Gamma^{*}, \Delta^{*} )}_{L}$ such that $\mathcal{L} \sqsubset L \sqsubset \mathcal{L^{+}}$, $\Gamma \subseteq \Gamma^{*}$ and $\Delta \subseteq \Delta^{*}$. 
    \item Let $\mathcal{L} \sqsubset L_{1} \sqsubset L_{2} \sqsubset \mathcal{L^{+}}$. If a sequent $\Gamma \Rightarrow \Delta$ in $L_{1}$ is not derivable in $\fljc$, then there exists a prime $L_{2}$-complete pair ${( \Gamma^{*}, \Delta^{*} )}_{L_{2}}$ such that $\Gamma \subseteq \Gamma^{*}$ and $\Delta \subseteq \Delta^{*}$.
\end{enumerate}
\end{lem}

\begin{proof}
Since the ways of showing (1) and (2) are very similar, the only former is described here. Suppose a sequent $\Gamma \Rightarrow \Delta$ in $\mathcal{L}$ is not derivable in $\fljc$. Consider a syntax $L$ such that $\mathcal{L} \sqsubset L \sqsubset \mathcal{L}^{+}$. Let $(A_{n})_{n \in \mathbb{N}}$ be an enumeration of all formulas in ${L}$. Since we will work only on this syntax, the suffix ``$L$'' to a pair of sets of formulas is omitted in the rest of this proof. We inductively define a pair ${( \Gamma_{n}, \Delta_{n} )}_{n \in \mathbb{N}}$ of sets of formulas such that a sequent $\Gamma_{n} \Rightarrow \Delta_{n}$ in $L$ is not derivable, $\Gamma_{n} \subseteq \Gamma_{n+1}$, and $\Delta_{n} \subseteq \Delta_{n+1}$ as follows:
\begin{description}
    \item [(Basis)] ${( \Gamma_{0}, \Delta_{0} )}$ = ${( \Gamma, \Delta )}$,
    
    \item [(Inductive Step)]
    Suppose ${( \Gamma_{i}, \Delta_{i} )}$ is already defined for any $i$ which satisfies $0 \leq i \leq n$. Let $z$ satisfy $z \in \mathsf{Var}(L)$ and $z \notin \mathsf{Var}(\mathcal{L})$ and do not occur free in $\Gamma_{n}$, $\Delta_{n}$, or $A_{n}$. Then ${( \Gamma_{n+1}, \Delta_{n+1} )}$ is defined as follows: 
    
    \[
    {( \Gamma_{n+1}, \Delta_{n+1} )} = 
    \begin{cases}
    {( \Gamma_{n} \cup \{ A_{n} \}, \Delta_{n})} & \text { if } \Gamma, A_{n} \Rightarrow \Delta \text{ is underivable and } A_{n} \not\equiv \exists x B, \\
    {( \Gamma_{n} \cup \{ \exists x B, B[z/x] \}, \Delta_{n})} & \text{ if } \Gamma, A_{n} \Rightarrow \Delta \text{ is underivable and } A_{n} \equiv \exists x B,\\
    {( \Gamma_{n}, \Delta_{n} \cup \{ A_{n} \})} & \text{ if } \Gamma \Rightarrow \Delta, A_{n} \text{ is underivable and } A_{n} \not\equiv \forall_{\mathtt{c}} x B,\\
    {( \Gamma_{n}, \Delta_{n} \cup \{ \forall_{\mathtt{c}} x B, B[z/x] \})} & \text{ if } \Gamma \Rightarrow \Delta, A_{n} \text{ is underivable and } A_{n} \equiv \forall_{\mathtt{c}} x B.\\
    \end{cases}
    \]
    It should be noted that a variable $z$ cannot be run out of, since the syntax $L$ is obtained by adding to $\mathcal{L}$ a new countably infinite set of variables. It can be checked that all of the four cases described above preserve underivability. It can also be checked that  whatever formula $A_{n}$ is, it can be distinguished into one of the four cases described above.
\end{description}
Based on the definition, the sets of formulas $\Gamma^{*}$ and $\Delta^{*}$ are defined as follows, respectively:
\begin{center}
     \[
     \Gamma^{*} := \underset{n \in \mathbb{N}}{\bigcup} {\Gamma_{n}},
     \quad
     \Delta^{*} := \underset{n \in \mathbb{N}}{\bigcup} {\Delta_{n}}.
    \]
\end{center}
It can be shown that the pair ${( \Gamma^{*}, \Delta^{*} )}_{L}$ satisfies the conditions required in this lemma. 
\end{proof}


\begin{definition}[Canonical Model]
\label{def:cm}
Define the canonical model $M^{c}$ = $(W^{c},R^{c},(D^{c}({( \Gamma, \Delta )}_{L}))_{{( \Gamma, \Delta )}_{L} \in W^{c}},V^{c})$ of a syntax $\mathcal{L}$ as follows:
\begin{itemize}
    \item $W^{c}$ := $\{ {( \Gamma, \Delta )}_{L} | \mathcal{L} \sqsubset L \sqsubset \mathcal{L^{+}}$ and ${( \Gamma, \Delta )}_{L}$ is a prime $L$-complete pair $\}$,
    \item ${( \Gamma, \Delta )}_{L} R^{c} {( \Gamma', \Delta' )}_{L'}$ iff all of the following hold:
    \begin{itemize}
        \item If $P(t_{1}, \ldots, t_{m}) \in \Gamma$ holds, then $P(t_{1}, \ldots, t_{m}) \in \Gamma'$ holds,
        \item If $A \imp_{\mathtt{i}} B \in \Gamma$ holds, then $A \imp_{\mathtt{i}} B \in \Gamma'$ holds, and
        \item If $\forall_{\mathtt{i}}x A \in \Gamma$ holds, then $\forall_{\mathtt{i}} x A \in \Gamma'$ holds,
    \end{itemize}
    \item $D^{c}(( \Gamma, \Delta )_{L})$ = $\{ t | t \text{ is a term of } L \}$,
    \item Define a valuation $V^{c}$ as follows:
\begin{itemize}
    \item $\langle V(t_{1}), \ldots, V(t_{m}) \rangle \in V(P,{( \Gamma, \Delta )}_{L})$  iff  $P(t_{1}, \ldots, t_{m}) \in \Gamma$,
    \item $V(c)$ := $c$, where $c$ is a constant in $\mathcal{L}$.
\end{itemize}
\end{itemize}
\end{definition}

\noindent It is easy to see that 
the canonical model $M^{c}$ is {\em well-defined}, i.e., $R^{c}$ is a preorder and $V^{c}$ satisfies heredity for predicates.
By the induction on the complexity of a formula, Lemma \ref{lem:tl} is shown, where Lemma \ref{lem:ccc} and Lemma \ref{lem:com} (2) enable us to deal with the cases when $A$ is of the form $B \imp_{\mathtt{c}} C$ and when $A$ is of the form $B \imp_{\mathtt{i}} C$ or $\forall_{\mathtt{i}}x B$, respectively. 

\begin{lem}[Truth Lemma]
\label{lem:tl}
For any ${( \Gamma, \Delta )}_{L} \in W^{c}$, any $L$-formula $A$, i.e., any formula $A \in \mathsf{Form}_{L}$, and any $\{ x_{1}, \ldots, x_{n} \} \subseteq \mathsf{Var}(L)$ such that $\mathsf{FV}(A) \subseteq \{ x_{1}, \ldots, x_{n} \}$, the following equivalence holds: 
\[
\begin{array}{lll}
     A \in \Gamma 
     &\mathrm{iff}& 
     {( \Gamma, \Delta )}_{L} \models_{M^{c}} A[\underline{x_{1}}/x_{1}] \cdots [\underline{x_{n}}/x_{n}],\\
\end{array}
\]
where it is noted that $A[\underline{x_{1}}/x_{1}] \cdots [\underline{x_{n}}/x_{n}]\in \mathcal{L}(D^{c}{( \Gamma, \Delta )}_{L})$. 
\end{lem}

\begin{proof}
We use induction on the complexity of a formula $A$. The case when the main connective of $A$ is intuitionistic implication and the case when it is intuitionistic universal quantifier are dealt with here. 
\begin{itemize}
    \item Let $A$ be of the form $B \imp_{\mathtt{i}} C$.
\begin{description}
    \item [(From left to right)] Suppose $B \imp_{\mathtt{i}} C \in \Gamma$. Our goal is to show that\linebreak ${( \Gamma, \Delta )}_{L} \models_{M^{c}} (B \imp_{\mathtt{i}} C)[\underline{x_{1}}/x_{1}] \cdots [\underline{x_{n}}/x_{n}]$, which is syntactically the same as\linebreak ${( \Gamma, \Delta )}_{L} \models_{M^{c}} B[\underline{x_{1}}/x_{1}] \cdots [\underline{x_{n}}/x_{n}] \imp_{\mathtt{i}} C[\underline{x_{1}}/x_{1}] \cdots [\underline{x_{n}}/x_{n}]$. Fix any ${( \Gamma', \Delta' )}_{L'} \in W^{c}$ such that ${( \Gamma, \Delta )}_{L} R^{c} {( \Gamma', \Delta' )}_{L'}$ and ${( \Gamma', \Delta' )}_{L'} \models_{M^{c}} B[\underline{x_{1}}/x_{1}] \cdots [\underline{x_{n}}/x_{n}]$. It suffices to show ${( \Gamma', \Delta' )}_{L'} \models_{M^{c}} C[\underline{x_{1}}/x_{1}] \cdots [\underline{x_{n}}/x_{n}]$. By applying induction hypothesis, $B \in \Gamma'$ is obtained. Since both $B \imp_{\mathtt{i}} C \in \Gamma$ and ${( \Gamma, \Delta )}_{L} R^{c} {( \Gamma', \Delta' )}_{L'}$ hold, $B \imp_{\mathtt{i}} C \in \Gamma'$ holds. Thus, both $\Gamma' \Rightarrow B$ and $\Gamma' \Rightarrow B \imp_{\mathtt{i}} C$ are derivable in $\fljc$. By applying to these sequents $(\imp_{\mathtt{i}} \Rightarrow)$, $(Cut)$, and $(c \Rightarrow)$, the sequent $\Gamma \Rightarrow C$ is derived. From the condition ($\Gamma'$ is a theory) in Definition \ref{def:prime}, $C \in \Gamma'$ holds, and by applying induction hypothesis, ${( \Gamma', \Delta' )}_{L'} \models_{M^{c}} C[\underline{x_{1}}/x_{1}] \cdots [\underline{x_{n}}/x_{n}]$ is obtained, as is desired. 

    \item [(From right to left)] Fix any ${( \Gamma, \Delta )}_{L} \in W^{c}$. In order to show the contrapositive, suppose $B \imp_{\mathtt{i}} C \notin \Gamma$. From the condition ($\Gamma$ is a theory) in Definition \ref{def:prime}, the sequent $\Gamma \Rightarrow B \imp_{\mathtt{i}} C$ is not derivable in $\fljc$. Consider the set $\Theta$ of all persistent formulas contained in $\Gamma$. Since $\Theta \subseteq \Gamma$, the sequent $\Theta \Rightarrow B \imp_{\mathtt{i}} C$ is not derivable in $\fljc$. This implies the sequent $B, \Theta \Rightarrow C$  is not derivable, because if this were, the sequent $\Theta \Rightarrow B \imp_{\mathtt{i}} C$ would be derivable by applying $(\Rightarrow \imp_{\mathtt{i}})$. Thus, from Lemma \ref{lem:com} (2), there exists a syntax $L'$ and a prime $L'$-complete pair $(\Gamma',\Delta')_{L'}$ such that $\mathcal{L} \sqsubset L \sqsubset L' \sqsubset \mathcal{L^{+}}$, $\Theta \cup \{ B \} \subseteq \Gamma'$ and $\{ C \} \subseteq \Delta'$. From the condition (underivability) in Definition \ref{def:prime}, $C \notin \Gamma'$. Since $\Theta$ contains all of the persistent formulas in $\Gamma$, ${( \Gamma, \Delta )}_{L} R^{c} (\Gamma', \Delta')_{L'}$ holds. By applying induction hypothesis to $B \in \Gamma'$ and $C \notin \Gamma'$, ${( \Gamma', \Delta' )}_{L'} \models_{M^{c}} B[\underline{x_{1}}/x_{1}] \cdots [\underline{x_{n}}/x_{n}]$ and ${( \Gamma', \Delta' )}_{L'} \not\models_{M^{c}} C[\underline{x_{1}}/x_{1}] \cdots [\underline{x_{n}}/x_{n}]$ are obtained. These ensures ${( \Gamma, \Delta )}_{L} \not\models_{M^{c}} (B \imp_{\mathtt{i}} C)[\underline{x_{1}}/x_{1}] \cdots [\underline{x_{n}}/x_{n}]$, as is desired.
\end{description}

    \item Let $A$ be of the form $\forall_{\mathtt{i}} x B$.
    \begin{description}
    \item [(From left to right)] Suppose $\forall_{\mathtt{i}}x B \in \Gamma$. Our goal is to show ${( \Gamma, \Delta )}_{L} \models_{M^{c}} (\forall_{\mathtt{i}} x B)[\underline{x_{1}}/x_{1}] \cdots [\underline{x_{n}}/x_{n}]$, which is syntactically the same as ${( \Gamma, \Delta )}_{L} \models_{M^{c}} \forall_{\mathtt{i}} x (B[\underline{x_{1}}/x_{1}] \cdots [\underline{x_{n}}/x_{n}])$. Fix any  ${( \Gamma', \Delta' )}_{L'} \in W^{c}$ such that ${( \Gamma, \Delta )}_{L} R^{c} {( \Gamma', \Delta' )}_{L'}$. It suffices to show for all $t \in D^{c}(( \Gamma', \Delta' )_{L'})$, ${( \Gamma', \Delta' )}_{L'} \models_{M^{c}} B[\underline{x_{1}}/x_{1}] \cdots [\underline{x_{n}}/x_{n}][\underline{t}/x]$. Since both $\forall_{\mathtt{i}} x B \in \Gamma$ and ${( \Gamma, \Delta )}_{L} R^{c} {( \Gamma', \Delta' )}_{L'}$ hold, $\forall_{\mathtt{i}} x B \in \Gamma'$ holds, which implies $\Gamma' \Rightarrow \forall_{\mathtt{i}} x B$ is derivable in $\fljc$. It is trivial that the sequent $B[t/x] \Rightarrow B[t/x]$. By applying to these sequents $(\forall_{\mathtt{i}} \Rightarrow)$ and $(Cut)$, the sequent $\Gamma' \Rightarrow B[t/x]$ is obtained. From the condition ($\Gamma'$ is a theory) in Definition \ref{def:prime}, $B[t/x] \in \Gamma'$ holds. The term $t$ can be a constant, one of $x_{1}, \ldots, x_{n}$, or a variable different from any of $x_{1}, \ldots, x_{n}$. The only last case is dealt with here. By applying induction hypothesis, ${( \Gamma', \Delta' )}_{L'} \models_{M^{c}} B[t/x][\underline{x_{1}}/x_{1}] \cdots [\underline{x_{n}}/x_{n}][\underline{t}/t]$ is obtained. 
    We note that $B[t/x][\underline{x_{1}}/x_{1}] \cdots [\underline{x_{n}}/x_{n}][\underline{t}/t]$
    is $B[\underline{x_{1}}/x_{1}] \cdots [\underline{x_{n}}/x_{n}][t[\underline{x_{1}}/x_{1}] \cdots [\underline{x_{n}}/x_{n}]/x][\underline{t}/t]$. Since $t$ is distinguished from any of $x_{1}, \ldots, x_{m}$, $B[\underline{x_{1}}/x_{1}] \cdots [\underline{x_{n}}/x_{n}][t[\underline{x_{1}}/x_{1}] \cdots [\underline{x_{n}}/x_{n}]/x][\underline{t}/t]$ is $B[\underline{x_{1}}/x_{1}] \cdots [\underline{x_{n}}/x_{n}][t/x][\underline{t}/t]$, which is syntactically the same as $B[\underline{x_{1}}/x_{1}] \cdots [\underline{x_{n}}/x_{n}][\underline{t}/x]$. This argument implies ${( \Gamma', \Delta' )}_{L'} \models_{M^{c}} B[\underline{x_{1}}/x_{1}] \cdots [\underline{x_{n}}/x_{n}][\underline{t}/x]$, as is desired. 

    \item [(From right to left)] Fix any ${( \Gamma, \Delta )}_{L} \in W^{c}$. In order to show the contrapositive, suppose $\forall_{\mathtt{i}}x B \notin \Gamma$. From the condition (underivability) in Definition \ref{def:prime}, the sequent $\Gamma \Rightarrow \forall_{\mathtt{i}}x B$ is not derivable in $\fljc$. Consider the set $\Theta$ of all persistent formulas contained in $\Gamma$. Since $\Theta \subseteq \Gamma$, the sequent $\Theta \Rightarrow \forall_{\mathtt{i}}x B$ is not derivable in $\fljc$. Consider a syntax $L_{1}$, obtained by adding to $L$ a variable $z$, which is not in $L$. The fact that the syntax $L$ is obtained by adding a new countably infinite set of variables to $\mathcal{L}$ ensures $\mathcal{L} \sqsubset L_{1}$. Since $z$ does not occur free in $\Theta$ and $\forall_{\mathtt{i}}x B$, the sequent $\Theta \Rightarrow B[z/x]$ is not derivable, because if this sequent were, the sequent $\Theta \Rightarrow \forall_{\mathtt{i}}x B$ would be derivable by applying $(\Rightarrow \forall_{\mathtt{i}})$. Thus, from Lemma \ref{lem:com} (2), there exists a syntax $L_{2}$ and a prime $L_{2}$-complete pair $(\Gamma',\Delta')_{L_{2}}$ such that $\mathcal{L} \sqsubset L_{1} \sqsubset L_{2} \sqsubset \mathcal{L^{+}}$, $\Theta \subseteq \Gamma'$ and $\{ B[z/x] \} \subseteq \Delta'$. By (underivability) in Definition \ref{def:prime}, $B[z/x] \notin \Gamma'$. Since $\Theta$ contains all of the persistent formulas in $\Gamma$, ${( \Gamma, \Delta )}_{L} R^{c} (\Gamma', \Delta')_{L_{2}}$ holds. By applying induction hypothesis to $B[z/x] \notin \Gamma'$, ${( \Gamma', \Delta' )}_{L_{2}} \not\models_{M^{c}} B[z/x][\underline{x_{1}}/x_{1}] \cdots [\underline{x_{n}}/x_{n}][\underline{z}/z]$ is obtained. We note that $B[z/x][\underline{x_{1}}/x_{1}] \cdots [\underline{x_{n}}/x_{n}][\underline{z}/z]$ is $B[\underline{x_{1}}/x_{1}] \cdots [\underline{x_{n}}/x_{n}][z[\underline{x_{1}}/x_{1}] \cdots [\underline{x_{n}}/x_{n}]/x][\underline{z}/z]$. Since $x$ is distinguished from any of $x_{1}, \ldots, x_{n}$, $B[\underline{x_{1}}/x_{1}] \cdots [\underline{x_{n}}/x_{n}][z[\underline{x_{1}}/x_{1}] \cdots [\underline{x_{n}}/x_{n}]/x][\underline{z}/z]$ is syntactically the same as  $B[\underline{x_{1}}/x_{1}] \cdots [\underline{x_{n}}/x_{n}][z/x][\underline{z}/z]$, which is $B[\underline{x_{1}}/x_{1}] \cdots [\underline{x_{n}}/x_{n}][\underline{z}/x]$. This argument implies ${( \Gamma', \Delta' )}_{L_{2}} \not\models_{M^{c}} B[\underline{x_{1}}/x_{1}] \cdots [\underline{x_{n}}/x_{n}][\underline{z}/x]$. The above argument ensures ${( \Gamma, \Delta )}_{L} \not\models_{M^{c}} (\forall_{\mathtt{i}}x B)[\underline{x_{i}}/x_{i}]$, as is desired.
    \qedhere
\end{description}
\end{itemize}
\end{proof}

Finally, we obtain the following desired theorem.

\begin{thm}[Strong Completeness Theorem]
\label{thm:sct}
For any set $\Gamma \cup \{ A \}$ of formulas in $\mathcal{L}$, if $\Gamma \models A$ then $\Gamma \Rightarrow A$ is derivable in $\fljc$.
\end{thm}

\begin{proof}
Suppose $\Gamma \Rightarrow A$ is not derivable in $\fljc$ in the original syntax $\mathcal{L}$ to show the contrapositive. 
By Lemma \ref{lem:com} (1), there exist a syntax $L$ and a prime $L$-complete pair ${( \Gamma', \Delta' )}_{L}$ such that $\mathcal{L} \sqsubset L \sqsubset \mathcal{L^{+}}$, $\Gamma \subseteq \Gamma'$, and  $\{ A \} \subseteq \Delta'$. By Definition \ref{def:cm}, the prime $L$-complete pair ${( \Gamma', \Delta' )}_{L} \in W^{c}$ holds. It is noted that $M^{c}$ {\em is} a Kripke model. 
By Lemma \ref{lem:tl}, $w \models_{M} C  [\underline{x_{1}}/x_{1}] \cdots [\underline{x_{m}}/x_{m}]$ holds for any $C \in \Gamma$, and $w \not\models_{M} A [\underline{z_{1}}/z_{1}] \cdots [\underline{z_{n}}/z_{n}]$ also holds, where $x_{1},\ldots, x_{m}$ are all free variables in $C$ and $z_{1}, \ldots, z_{n}$ are all free variables in $A$. This concludes that $\Gamma \not\models A$.
\end{proof}

\section{Conclusion and Further Direction}
\label{sec:caf}

We have four further directions to do in the future. Firstly, a Hilbert-style axiomatization (which is equipollent with $\fljc$) should be provided. We try to do this based on a Hilbert-style axiomatization proposed in~\cite{Cerro1996}, which was shown to be equipollent with propositional $\ljc$ in~\cite{Toyooka2021a}.

Secondly, we may establish Craig interpolation theorem for $\fljc$. The interpolation theorem in an ordinary first-order intuitionistic and classical logic was shown in~\cite{Troelstra2012}, but whether the theorem holds in $\fljc$ is not solved yet. Since two implications (intuitionistic and classical ones) exist in $\fljc$, the corresponding two types of the theorem can be considered. We have already shown these two theorems for propositional $\ljc$ syntactically in~\cite{Toyooka2021a} by employing the idea in~\cite{Maehara1961}.

Thirdly, strong completeness of the cut-free fragment of $\fljc$ should be interesting, because such completeness will ensure the semantic proof of the cut elimination theorem. Then, we need to check if our basic idea for the strong completeness could be done without the notion of $L$-completeness in Definition \ref{def:prime}. Such an attempt was already done to a single-succedent sequent calculus for first-order intuitionistic logic in~\cite{Hermant2005,Mints2000}, but we will deal with multi-succedent $\fljc$. 

Fourthly, comparison of $\fljc$ with $\mathbf{LEci}$, which is a sequent calculus for Prawitz's ecumenical system, should be done from a proof-theoretic viewpoint. The calculus $\mathbf{LEci}$, which is mentioned in Section \ref{sub:rw}, was provided in~\cite{Elaine2019} in the $\mathbf{G3}$ style, which has no structural rule. Therefore, what is desired is providing a sequent calculus for the system which has structural rules, since such a calculus enables the comparison of ecumenical system with $\fljc$. It should be noted that $\fljc$ and $\mathbf{LEci}$ are not equivalent. This is because $A \imp_{\mathtt{i}} (B \imp_{\mathtt{i}} A)$ is not derivable in $\fljc$, as is explained in Section \ref{subsec:semex}, but is derivable in $\mathbf{LEci}$. From a model-theoretic viewpoint, in the Kripke semantics for $\mathbf{LEci}$, all formulas satisfy heredity. 

\bibliographystyle{eptcs}
\bibliography{biblio}

\end{document}